\documentclass[conference,10pt]{IEEEtran}   

\usepackage{xy}
\usepackage{amsthm}
\usepackage{algorithm}
\usepackage{algpseudocode}
\usepackage{textcomp}
\usepackage{amsmath}
\usepackage{stmaryrd}
\usepackage{graphicx}
\usepackage{syntax}
\usepackage{fancyhdr}
\usepackage{textcomp}
\usepackage{flushend}
\usepackage{semantic}
\usepackage{marvosym}
\usepackage[short,nodayofweek,level,12hr]{datetime}  

\theoremstyle{plain} 
\newtheorem{thm}{Theorem}[section]

\newtheorem{theorem}[thm]{Theorem}
\newtheorem{definition}[thm]{Definition}

\newtheorem{lemma}[thm]{Lemma}
\newtheorem{conventions}[thm]{Conventions}
\newtheorem{observation}[thm]{Observation}
\newtheorem{example}[thm]{Example}

\begin{document}
\xyoption{all}

\title{A Non-repetitive Logic for Verification of Dynamic Memory with Explicit Heap Conjunction and Disjunction}

\author{
  \IEEEauthorblockN{Ren\'{e} Haberland \qquad Kirill Krinkin}
  \IEEEauthorblockA{Saint Petersburg Electrotechnical University "\textit{LETI}"\\
  Saint Petersburg, Russia\\
  email: haberland1@mail.ru, kirill.krinkin@fruct.org}

}

\maketitle

\begin{abstract}
In this paper, we review existing points-to Separation Logics for dynamic memory reasoning and we find that different usages of heap separation tend to be an obstacle. Hence, two total and strict spatial heap operations are proposed upon heap graphs, for conjunction and disjunction -- similar to logical conjuncts. Heap conjunction implies that there exists a free heap vertex to connect to or an explicit destination vertex is provided. Essentially, Burstall's properties do not change.
By heap we refer to an arbitrary simple directed graph, which is finite and may contain composite vertices representing class objects. Arbitrary heap memory access is restricted, as well as type punning, late class binding and further restrictions. Properties of the new logic are investigated, and as a result group properties are shown. Both expecting and superficial heaps are specifiable. Equivalence transformations may make denotated heaps inconsistent, although those may be detected and patched by the two generic linear canonization steps presented. The properties help to motivate a later full introduction of a set of equivalences over heap for future work.
Partial heaps are considered as a useful specification technique that help to reduce incompleteness issues with specifications. Finally, the logic proposed may be considered for extension for the Object Constraint Language.
\end{abstract}

\begin{flushleft}
\end{flushleft}

\IEEEpeerreviewmaketitle

\section{Introduction}
\label{section:Introduction}
In contrast to automatically allocated memory, which remains in the stack, \textit{dynamic memory} refers to the main memory part that is altered by commands such as \texttt{new}, \texttt{delete} and heap data assignments. The dynamic memory contains \textit{heaps} (see definition \ref{def:ReynoldsHeapDefinition}). Let us first review a few important definitions and discuss issues with heaps afterwards.

Jones  et al. \cite{Jones11} define a \textit{heap} as a contiguous subscripted datastructure, and also, alternatively, as an organised graph of ``\textit{discontiguous blocks of contiguous words}''. All allocated memory cells have a reference and the liveness of a cell is defined by its reachability. The liveness is independent from the program statement that creates a dynamic memory cell.

Reynolds \cite{reynolds02} defines \textit{heaps} (not to be mixed up with a single heap) as the union of all mappings of address subsets to non-empty value cells. Following this definition a single heap would be some addresses pointing to some arbitrary data structure. Reynolds mentions his intention goes back to Burstall's model \cite{burstall72}. Both refer to trees as implied data structure - which, at least in Burstall's proposition, denotes a simple \textit{heap graph} (definition \ref{def:finiteHeapGraphDefinition} formally introduces it, for the moment let us assume it is an arbitrary graph where edges represent some relationship between heap vertices) as for instance the expression \xymatrix{ x \ar[r] ^{a_1,a_2,a_3} & y} denotes some path within the heap graph in Fig. \ref{fig:SampleHeapGraph1}. The graph starts at $x$ and stops at a heap cell which is also pointed by some local variable $y$ by visiting $a_1$, $a_2$, $a_3$, which all may have some unspecified content on its way there. Reynolds introduces the ``$\star$''-operator for expressing that two heaps do not share common dynamic memory regions. In contrast to Burstall Reynolds' model is slightly different: all except the start of a path from a stacked variable denotes its value rather its cell location. As shown in the graph in Fig. \ref{fig:HeapGraphCactus} by convention it is agreed that stacked variables, such as $x_1$,  only have outgoing edges, where all other vertices, such as $x_2, x_3, x_4, x_5, x_6, x_7$, denote some concrete heap cell values and may have zero or more incoming and zero or more outgoing edges.

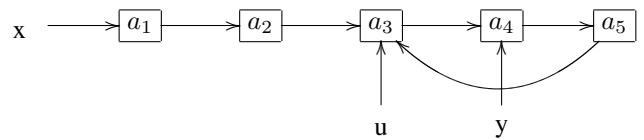
\begin{figure}[b]
\begin{xy}
\xymatrix @W=1.3pc @H=1pc @R=0pc{
  \txt{x} \ar[r] &   *+[F]\txt{$a_1$} \ar[r] & *+[F]\txt{$a_2$} \ar[r] & *+[F]\txt{$a_3$} \ar[r] & *+[F]\txt{$a_4$} \ar[r] & *+[F]\txt{$a_5$} \ar@/^2pc/[ll]\\\\
  &&& \txt{u} \ar[uu] & \txt{y} \ar[uu]
}
\end{xy}
\caption{An arbitrary annotated heap graph with locals $x$, $u$ and $y$ pointing to cells with some content $a_1$, $a_3$ and $a_4$}
\label{fig:SampleHeapGraph1}
\end{figure}

\begin{figure}[t]
\begin{displaymath}
 \xymatrix{
              &             & x_3 \ar[r] & x_4 \ar[dr] \\
   x_1 \ar[r] & x_2 \ar[ur] & x_5 \ar[r] & x_6 \ar[r] & x_7
 }
\end{displaymath}
\caption{Heap graph sample consisting of two simply linked lists forming inverted cactuses}
\label{fig:HeapGraphCactus}
\end{figure}
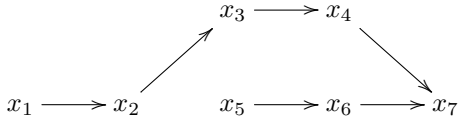

If we like to parameterise a heap graph so it contained \textit{genuine symbolic variables}, we rather have to distinguish between parameterised and fixed variable meanings on each verification step.
Reynolds introduces the ``,''-operator to specify paths in heap graphs. For example, when using ``,'' the above data-structure could be fully specified by $x_1 \mapsto x_2,x_3,x_4,x_7 \wedge x_5 \mapsto x_6,x_7$. For comparison, the same data structure without the path-operator ``,'' would be $(x_2 \mapsto x_3 \star x_4 \mapsto x_7 \star x_5 \mapsto x_6) \wedge (x_1 \mapsto x_2 \star x_3 \mapsto x_4 \star x_6 \mapsto x_7)$ -- we excuse ourselves variable locations and content were mixed up in this example for the sake of simplicity. Based on the ``$\star$''-operator and ``$\mapsto$'' the so-called \textit{Separation Logic} was proposed \cite{reynolds02},\cite{parkinsonThesis05} and implemented \cite{berdineCO05b}.
The following example \cite{reynolds02} demonstrates the definition of a binary tree predicate (we call a predicate ``\textit{abstract}'' whenever it depends on parameters):

$tree(l)::=\texttt{nil} \ | \ \exists x.\exists y:\ l \mapsto x,y \ \star \ tree(x) \ \star \ tree(y)$

The abstraction parameter $l$ in Fig. \ref{fig:HeapGraphTree} is some variable symbol and $tree$ denotes the recursively defined predicate implying the left branch does not intersect with any part of the right branch, and vice versa. However, strictly speaking this must not always be the case, since in the above specification $tree(y)$ might be substituted by $tree2(x,y)$, which could hypothetically link back to $x$ again and so would breach the convention made previously -- luckily, this can be excluded in most cases, except when references to dynamic memory are determined on runtime. For example, \texttt{p[13+offset]} where \texttt{offset} is decidable on runtime only might be such a scenario. The breach may be avoided for $tree2$ just by not passing $x$ neither recalling it globally.  Even if the tree entirely fits into dynamic memory, remember $x$ and $y$ get stacked once the tree is traversed: first $x$, then $y$ is accommodated at the next available address because of ``,''. The authors of this paper are aware of dropping unbound heap memory access may induce considerable practical restrictions, however, we think this restriction can in many cases be overcome by a modification to the chosen data model.

By convention, whenever a vertex of the heap graph has at most one outgoing edge, the heap graph is \textit{simple}, e.g. linearly-linked lists, trees and arbitrary heap graphs without multiple edges between two vertices.  W.l.o.g. we consider only pointers that refer to particular heap cells or class objects that may union several pointers to heap cells. We will further also consider abstract predicates. In order to decide whether two heaps indeed do not share a common heap, it is necessary to check there exists no path from one heap graph to the other.

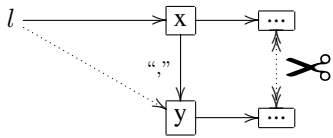
\begin{figure}[b]
\begin{displaymath}
 \xymatrix{
  l \ar[rr] \ar@{.>}[drr] && *+[F] \txt{x} \ar[d]_{``,"} \ar[r] &  *+[F] \txt{...}  \\
                                    && *+[F] \txt{y} \ar[r] & *+[F] \txt{...} \ar@{|<<.>>|}[u]_{\txt{\Huge{\Leftscissors}}}
 }
\end{displaymath}
\caption{Heap graph instance for the definition of a binary $tree$. The left $x$ child points to some content which may not interfere with the content pointed to by $y$}
\label{fig:HeapGraphTree}
\end{figure}

One alternative to Separation Logic is \textit{Shape Analysis} \cite{sagiv02}. It makes use of transfer functions in order to describe changes to the heap by every program statement. Another approach, as being demonstrated by Baby Modula 3 \cite{abadi93}, uses a class-free object calculus and a single unique result register. This register stores the result after each single statement and so allows to refer to the state before and after running a particular statement. Class-objects and their typeable theories are discussed in \cite{abadi97},\cite{abadi93}. 

At this point it is worth noting that a points-to model is considered in this paper due to its \textit{locality property} w.r.t. heap graphs, modifications do not usually require a full specification update due to its edge-based graph representation.

The inspiration for this paper, even if coming from a different context, is \cite{suzuki82}, where a rather intuitive but incomplete set of ''\textit{safe}`` symmetry operations on pointers is proposed in order to prove correctness of more complex pointer manipulations. Safe operations, as rotation or shift, raise big practical concerns as hard-to predict pointer behaviour even on very small modifications as well as incompleteness gaps on pointer rotations.

The main purpose of this paper is to present two new \textit{context-free} binary operations for heap conjunction and heap disjunction, to show group properties hold and those can be used for example for proof simplifications on proof rules in the future.

Section \ref{section:Introduction} of this paper gave a brief overview of the topic and related problems. Section \ref{section:HeapSeparatingLogic} introduces briefly Separation Logic, it introduces a concluded definition of heap and heap graph, and it comes with conventions for class objects and heap memory alignment. The main part, section \ref{sec:ConjunctionAndDisjunction} defines heap terms to be interpreted within heap formulae. Pointers of pointers and arrays are only very briefly discussed, heap conjunction is introduced for basic (''\textit{heaplet}``) and generalised heaps (what is later expressed as heap term) as well as path accessors (see observation \ref{obs:ObjectPathAccessor}). Properties of the conjunction are investigated and established, canonisation steps are demonstrated in order to overcome transient inconsistency, which may occur from references no more alive. Heap inversion is proposed as notational trick. In companion with other properties it may eventually help to define equalities over heaps and so improve the comparison with expected heaps in the future. In particular, defining a partially-ordered set over conjunct heaps w.r.t. sub-graph inclusion, and distributivity along with other properties would eventually help to define for instance a \textit{satisfiability-modulo-theory}. Partial specification is introduced in section \ref{sec:PartialSpec} for objects. Discussions propose an extension with aliases to the \textit{Object Constraint Language}. Finally, conclusions follow a short discussion.


\section{Heap Separating Logic}
\label{section:HeapSeparatingLogic}

Before going into more detail, let us first ask whether we cannot simply turn all dynamically allocated variables into stacked, as it was proposed, for instance, by \cite{meyer1-03}. Often this will indeed work, however, sometimes it is not a good idea after all, because of performance issues \cite{appel87}, for instance. More often the nature of the problem forbids general static assumptions on stack bounds. An overview and numerous definitions of dynamic memory may be found in \cite{Jones11}.

The essence of Reynolds' heap model and properties was briefly wrapped up in the previous section. So, one central problem seems to be expressibility, which is the main purpose of this paper. This section introduces a heap by referring to a graph, followed by numerous model discussions and property observations.\\

\begin{definition} (concluded from Reynolds) A heap is defined as $\bigcup_{A \subseteq Addr} A \mapsto Val^n$ with $n\ge 1$, $A$ being some address set and $Val$ is some value domain, for instance, integers or inductively defined structures containing $A$. A heap may be composed inductively by other heaps as following:

$H_1 \star H_2$, where $H_1$ describes some heap graph assertion $H_1=(V_1,E_1)$ and in analogy to that $H_2=(V_2,E_2)$, where edges $E=V \times V$ and edges are directed, s.t. iff $\forall v_1 \in V_1$, $v_2 \in V_2$ with $v1 \ne v2$ and cases:

\begin{itemize}
 \item[1st] (\textit{Separation}): $(v_1,v_2) \notin E_1$, and $(v_1,v_2) \notin E_2$.
 \item[2nd] (\textit{Conjunction}): $\exists s \in V_1,\exists t \in V_2: (s,t) \in E_1$ or $(s,t) \in E_2$ then $H_1$ or $H_2$ contains some $\star$-separated $s \mapsto t$.
\end{itemize}
\label{def:ReynoldsHeapDefinition}
\end{definition}

Variables as well as pointers are stored in the stack, where the content pointed to remains in heap memory (the following domain equation \cite{berdineCO05b} holds: $Stack=Values \cup Locals$). Heap graph assertions are assertions about the heap graph constructed by program statements manipulating the dynamic memory. Those assertions are interpreted as \textit{true} or \textit{false} depending on whether an associated concrete heap corresponds or not. Definition \ref{def:HeapTermExtendedDefinition} will introduce the syntax for heap assertions.

Regarding definition \ref{def:ReynoldsHeapDefinition} the overloading of the binary operator ''$\star$`` happens in two ways: one is to express two heaps do not overlap, and the second way is to express two heaps are somehow linked together by using transient symbols. The ''$\star$``-operator is a logical and spatial conjunction, it links two prepositions about heaps together and it describes heap entities which have some configuration in space, both consume different dynamic memory regions. On the one hand, if we link strictly two separate heaps then we have to find a maximal matching in order to describe the given \textit{heap graph} entirely, which is impractical. On the other hand, separation also seems to be a very elegant way to separate specification concerns locally: if there is an assertion regarding a particular data structure in heap, this should involve at most only that data structures and exclude unaffected ones. After all, the above initial definition seems complicated enough, because it is ambiguous and it refers to a single heap definition, which should ideally not be too different from Reynolds' initial and rather intuitive definition of heaps -- but as we have seen unfortunately, it is.
So, two strict operators would be desired rather a single ''$\star$``, one operator to strictly separate and one to join heaps. Heaps shall be replaceable with symbolic placeholders in order to beat ambiguity whilst analysing verification conditions. Moreover, syntax and semantic intention of heap expressions shall be unified.

Once both heap operators are defined, properties and equivalences can be established separately. Finally, heap theories and term algebras may eventually be proposed in future over both heap operators. In definition \ref{def:finiteHeapGraphDefinition} we first need to formally define what a heap actually is.

The underpinning theory behind \cite{burstall72} is the so-called \textit{Substructural Logic} \cite{dosen93}, which is a higher-order logic, a logic where, for instance, the contraction rule does no more hold, constants have in fact turned into predicates that may be quantified (details can be found in \cite{dosen93}). Contraction-freeness \cite{restall94} in this context has for our purpose the advantage of non-repeating heap entities within a heap assertion. By repeating we directly refer to points-to expressions as defined later.

\begin{definition}
A \textit{(finite) heap graph} is a directed connected graph within the dynamic memory section which may contain cycles, but must remain simple. Each vertex has a type-dependent size and an unique memory address, but may not overlap with other vertices. Every edge represents a relationship, for instance, a pointer to some absolute memory address or a relative jump field displacement.
\label{def:finiteHeapGraphDefinition}
\end{definition}

The absolute addresses are out of interest to the verification. The heap graph must be pointed by at least one stacked variable, otherwise the so-defined graph is considered as garbage. Stacked variables may also point to one vertex, in this case all except one variable are \textit{aliases} of the variable considered.

The emphasis lies on finite, since only arbitrary big but finite address space shall be considered. The dynamic memory shall be linearly addressable, however some operations \texttt{new} and \texttt{delete} shall organise themselves how and where to allocate or free heap memory. We restrict ourselves pointers address \textit{correctly} and \textit{sound}, and for the purpose of this paper we neither care too much about an optimal memory coalescing strategy to pointers that is most likely expanded on runtime, nor primarily about garbage collection issues. What we concern about is only that there is a relationship between a pointer and a pointed-by region within the heap memory region, it does not even require a pointer contains an absolute address within the dynamic memory range as it is not the case with bi-directional XOR-calculated jump-fields, which are relative pointer offsets depending on the address provided ''somehow`` by an actual vertex address.

\begin{figure}[b]
\begin{center}
\begin{tabular}{ccc}
 \includegraphics[width=1.7cm]{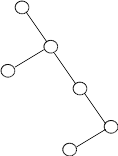} && \includegraphics[width=1.7cm]{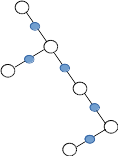}\\
 a)   && b) 
\end{tabular}
\end{center}
\caption{Schematic heap graph a) without direction, b) midpoints represent the substituted graph obtained by encoding source and destination vertices}
\label{fig:InvertedVertexGraph}
\end{figure}

\begin{conventions} Objects are restricted w.l.o.g. to be
\begin{itemize}
 \item[a)] non-inner objects only. Inner objects may always be modelled as with associated outer objects, so that there are references to different locations rather than all objects being accommodated within one contiguous heap chunk.
 
 \item[b)] Object fields and method names have to be all unique w.r.t. to its visibility. W.l.o.g. clashing names, for instance inherited names, are resolved by mangling the origin name with visibility mode and information from the deriving class. All references to mangled names need to be taken into account. This task is primarily done during the semantic program analysis phase.
 
 \item[c)] Due to encapsulation, objects do not grow in size normally, and due to late binding object references may keep invariant. However, the size of an object may spontaneously change. Sub-class objects may suddenly grow, but they may also shrink, depending on whether the translating compilation phase does align memory for non-used fields or not. If choosing a forced stack-allocated memory alignment for objects, then an object which is bound lately and passed alone to some procedure may better be reordered within its memory chunk, s.t. the growing part rises upwards on the stack.
 Because the separating heap are non-contractive \cite{burstall72}, object fields specified once may not appear again within a conjuncted heap expression.
  
 \item[d)] Arrays as base type are currently ignored. Multiple edges between the same two vertices are disallowed.
 
 \item[e)] Sharing of same heap cells by multiple cells is allowed to all object fields.
\end{itemize}
\label{conv:RestrictedObjects}
\end{conventions}

In order to stay consistent with the following definitions, a simple check for the incidence relationship for memory cells for a given heap graph needs to be introduced. A given heap graph is composed of points-to heaplets meaning a directed edge represents a location points to a heap address which contains some value. The check requires all these heaplet conjuncts are traversed and the desired heap graph shall be in an edge-centric representation.
However, whenever we like to determine if two heap vertices are incident with each other or not, we prefer a vertex-centric model encoding edges. So, we define the following built-in predicates: \texttt{reach\-es(x,y)}, \texttt{reach\-es(x,Y)}, \texttt{reach\-es(X,y)}, \texttt{reach\-es(X,Y)}, where \texttt{x} is a vertex and \texttt{X} denotes a vertex subset of a given heap graph (\texttt{y}, \texttt{Y} in analogy).

Both interleaved representations in Fig. \ref{fig:InvertedVertexGraph} (the vertex-centric representation is marked by smaller filled midpoints on every edge) are dual and convertible to each other. Midpoints encode \textit{source} and \textit{destination} as one vertex and link with former neighbouring vertices. Naturally, this mapping is bijective (proof skipped). Since we in general need to interpret predicates of at least first order, we could do this now by describing one heap graph by one conjunction of ''$loc \mapsto val$`` pairs rather than more complex forms.

\begin{conventions} (Heap Alignment)
\textit{Object fields} do not overlap, fields have distinguishing memory addresses. Pointers to objects and its fields may alias. An object is expressed commonly by the points-to expression $x \mapsto object(fld_1,fld_2, ...)$. It is agreed w.l.o.g. that object fields may not be accessed via arithmetic displacements but only by a valid object access path using the ''.``-operator and valid subordinate fields. W.l.o.g, but still for the sake of full computability late binding is skipped. A full support would imply the use of only the weakest common heap to be chosen.
\label{conv:HeapAlignment}
\end{conventions}


\section{Conjunction and Disjunction}
\label{sec:ConjunctionAndDisjunction}

Because of definitions \ref{def:ReynoldsHeapDefinition}, \ref{def:finiteHeapGraphDefinition} and conventions \ref{conv:RestrictedObjects}, \ref{conv:HeapAlignment} we describe a heap now by a term as following.

\begin{definition}
A \textit{heap term} describes heap graphs and is syntactically defined as:

\begin{tabular}{lll}
 $T::=$ & $loc \mapsto val$ \qquad \qquad & ... \textit{points-to heaplet} \\
        & $| \ T \circ T$ \qquad \qquad & ... \textit{heap conjunction} \\
        & $| \ T \ \| \ T$ \qquad \qquad & ... \textit{heap disjunction}\\
        & $| \ \underline{true} \ | \ \underline{false} \ | \ \underline{emp}$ & ... \textit{partial heap spec}\\
        & $| \ ( \ T \ )$ \qquad \qquad & ... \textit{subterm expression}
\end{tabular}

where $loc$ denotes a variable location, a location of a compound object field or a symbol representing some heap variable, and $val$ denotes some value of any arbitrary domain. $T$ describes the current heap state.
\label{def:HeapTermDefinition}
\end{definition}

$T$ can be considered as a formula since we do not restrict ourselves in not considering variable scopes as long as the syntax definition is obeyed. We further agree on that $\circ$ has lower precedence than $\|$, so $a_1 \circ a_2 \| a_3 \equiv (a_1 \circ a_2) \| a_3$. For the sake of notational simplicity, we do refer to $loc \mapsto val$, which besides is closer to Reynolds' definition rather than Burstall's. However, we really should better refer in practice to the address of the content being pointed to rather than the direct domain value, which naturally may be composed. Hence, we agree without any further notice on some polymorphic notational helpers, like $address(val)$, which will allow us to address given values.

$\underline{true}$ implies certain (remaining) heap term(s) is a tautology, regardless of the actual term(s). In analogy to that stands $\underline{false}$, which implements a contradiction. $\underline{emp}$ is a constant built-in predicate implying a given heap must be empty to be satisfiable. This is why all three may be used to consume all not explicitly listed $\circ$-conjuncted heaplets.
The partial specification we get allows us to keep heap formulae simple, since we now may implicitly include all unaffected, but still intended, heaps belonging together. Partial specifications together with abstract predicates raise modularity. 
This becomes particularly of interest for class objects, where all field heaplets generate its own heaplet scope, which is different from the global non-object scope (see convention \ref{conv:HeapAlignment}). In fact we just discussed extensions to our heap term definition, which ought to be summarised: 

\begin{definition}
Extended heap terms $ET$ are heap terms with constant formulae, negation and logical conjuncts:

\begin{tabular}{lll}
 $ET::=$ & $T$ \qquad \qquad & ... \textit{heap term}\\
         & $| \ p(\alpha)$ \qquad \qquad & ... \textit{abstract predicate call}\\
         & $| \neg ET$ & ... logical negation\\
         & $| \ ET \wedge ET$ & ... logical conjunction\\
         & $| \ ET \vee ET$ & ... logical disjunction
\end{tabular}
\label{def:HeapTermExtendedDefinition}
\end{definition}

The logical conjunctions do not really require more explanations than already said. A predicate call to a previously defined predicate may invoke all dependent subcalls, although predicate calls are not classic procedure calls. An brief introduction of Prolog using predicates and reasoning a specific Hoare-calculus may be found in \cite{haberland14}. Predicates may be parameterised by zero or more heap terms bound to the predicate. Heap terms may be used as input or output terms, or even both at the same time.
Intentionally, heap terms are used as sub-expressions in logical assertions.
The verification of a predicate retrieves a Boolean value depending on if a given formula is obeyed.

\begin{observation}
\textit{Pointers of pointers} are syntactic sugar. They do not fundamentally extend the expressibility of heap graph assertions. Their only purpose w.r.t. heap terms is to have an additional indirection level for increased programming language flexibility. They act as placeholder or symbol variable for pointer locations.
\end{observation}

By pointers of pointers neither the heap graph itself gets extended nor the referenced heap in comparison with no pointers of pointers. Symbolic variables and placeholders are useful, because they may select numerous heaplets at once. But, the '',``-operator can do this already for linearly-linked lists and this operator was found superficial in terms of expressibility. In addition to that, it should be noted, that abstract predicates may also perform a selection of numerous heap cells.
Although not too useful in a theoretic manner, the above conjecture does not necessarily exclude usability gains in practice.\\

\begin{definition}
\textit{Heap conjunction} $H \circ \alpha \mapsto \beta$ is defined as heap graph, where $G=(V,E)$ is $H$'s heap graph representation, $H$ is a heap graph, and $\alpha \mapsto \beta$ is a points-to heaplet:

$\left\{
\begin{array}{lll}
  (V \cup \{\alpha,\beta\} \cup \beta', && \mbox{if } isFreeIn(\alpha,H) \\
   E \cup \{(\alpha,\beta)\} \cup \{ (\beta,b) | b \in \beta'\})  && \mbox{if } H=\underline{emp}\\
								  && \ \ \ (V=E=\emptyset)\\
  \mbox{\underline{false}}  && \mbox{otherwise}
\end{array}
\right.
$

where $\beta' = vertices(\beta) \subseteq V$ determines all heap graph vertices being directly pointed by $\beta$, which in case $\beta$ is an object includes all of the fields pointing to some vertices. Since $\alpha$ must be a unique location (for instance an object access path) there may be only either one or no heap vertex matching in $isFreeIn$ for a given heap $H$. The assumption is there is always exactly one matching free vertex for conjunction when building up a heap graph from a scratch, otherwise two heaps are not linkable.
\label{def:HeapConjunctionDefinition}
\end{definition}

\begin{figure}[t]
\begin{tabular}{c|c}
 $\xymatrix{
    \circ \ar[r]^a & \circ \ar@{..>}[d] & \circ \ar[r]^c    & \circ \\
                   & \circ \ar[r]^b     & \circ \ar@{..>}[u]
 }$
 &
 $\xymatrix{
    \circ \ar[r]^a & \circ \ar[r]^b & \circ \ar[r]^c    & \circ
 }$\\\\
 heap graph before ... & and after conjunction
\end{tabular}
\caption{Heap graph before and after heap conjunction.}
\label{fig:HeapGraphConjBeforeAfter}
\end{figure}
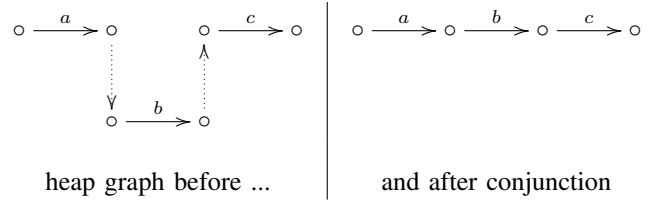

Lets say we would like to join three points-to pairs $a,b,c$ together (see Fig. \ref{fig:HeapGraphConjBeforeAfter}). First, $a$ must be expressed either purely by $a$ itself or by $\underline{emp} \circ a$. Only then $b$ might be connected to $a$, iff additionally $a$ contains actually a matching destination vertex that is not being assigned elsewhere. Once we have $a \circ b$, only then the edge $c$ may be connected as announced in the previous step, and we finally obtain the heap graph as seen in Fig. \ref{fig:HeapGraphConjBeforeAfter}. Since we may also conjunct any kind of graphs, e.g., binary trees, we allow to vary the conjunction ordering as long as the resulting graph is connected. For instance, $a \mapsto 5$ would be a valid heap conjunction, but $a \mapsto 5 \circ b \mapsto 5$ would be not. Notice that this way we still may express aliases if we want, for instance the heap graph $\xymatrix{ x \ \circ \ar[r] & \circ^z &  \ar[l] \circ \ y }$ could be expressed as heap term $x \mapsto z \circ y \mapsto z$.

In Fig. \ref{fig:HeapGraphConjBeforeAfter} all source and target vertices are simple and not annotated. In general the vertices may naturally be simple or compound in case of objects.
For the sake of simplicity, only a one-to-one connection is considered further, however assigning multiple objects to fields at once might be a very convenient method as long as the assignment is unambiguous, especially when it comes to arrays where a separator might be needed.\\
\textbf{Remark:}
Let $\varPhi_0$ be some heap, then $\varPhi = \varPhi_0 \circ a_0 \mapsto b_0 \Leftrightarrow \exists (a_m \mapsto b_m) \in \varPhi_0 \wedge (a_m=a_0, b_m\neq b_0 \vee a_m \neq a_0, b_m=b_0)$.

\begin{theorem}
 $H_1 \circ H_2$ conjuncts two heaps $H_1$ and $H_2$, if there exists at least one common vertex in each heap graph representations. It is agreed by convention $H_1 \circ \underline{emp} = \underline{emp} \circ H_1 = H_1$ holds.
\label{the:GeneralizedHeapConjunctionTheorem}
\end{theorem}
\begin{proof}
 This theorem is actually a generalisation of definition \ref{def:HeapConjunctionDefinition}. In contrast to definition \ref{def:HeapConjunctionDefinition}, the term on the right-hand side of the generalisation of $\circ$ is searched for a matching vertex. $H_1 \circ H_2$ represents one connected heap graph. Both, $H_1$ and $H_2$ may either be heaplet or a composition of heaplets of kind $a_1 \mapsto b_1 \circ a_2 \mapsto b_2 \circ \cdots \circ a_n \mapsto b_n$. In order to show the correctness of the theorem first we need to show is that if there is no common element in both heap graphs, then by definition \ref{def:HeapConjunctionDefinition} we obtain $\underline{false}$, which corresponds to what we would obtain from a conjunction. Otherwise, if we do have at least one common element, then inductively we do not bother about further common elements. So, we link both heap graphs up and the conjunction on heaps refers to connectible heaps. Further common elements would meld existing heap graph edges, the melded graph still is simply linked (but with multiple bridges), otherwise this would mean at least source or target of a melded heap graph edge would exclusively be either in $H_1$ or in $H_2$, and in both $H_1$ and $H_2$ at the same time, which is a contradiction, hence we just showed the theorem holds.
\end{proof}

Regarding the search for a matching element the $a\circ a$-decider mentioned later will resolve this practical issue.

\begin{observation}
 In an abstract predicate locations may be symbols. In order to increase reuse of abstract predicates for different locations and different location kinds (primarily for locals and object fields) it is agreed, that the field accessor ''.`` is a left-associative binary operator.
\label{obs:ObjectPathAccessor}
\end{observation}

Left-associativity means $object1.field1.field2.field3$ is internally represented by $(((object1).field1).field2).field3$. This way object access paths may be substituted by variable symbols, which raise modularity and flexibility of access paths expressions.

\begin{lemma}
$G=(\Omega, \circ)$ is a monoid, where $\Omega$ denotes the set of heap graphs and $\circ$ denotes heap conjunction.
\label{lem:HeapConjunctionMonoid}
\end{lemma}

\begin{proof}
In order to prove $G$ is a monoid we need to show (i) $\Omega$ is closed under $\circ$, (ii) $\circ$ is associative, and (iii) $\exists \varepsilon \in \Omega. \forall m \in \Omega: m \circ \varepsilon = \varepsilon \circ m = m$.

First of all, by $\omega \in \Omega$ we refer to a connected heap graph over ''$\mapsto$``-heaplets as being introduced in definition \ref{def:HeapTermDefinition}. According to definition \ref{def:HeapConjunctionDefinition} $\forall \omega \in \Omega: \omega \circ \omega = \underline{false}$, which is defined. Alternatively, there may be only two cases for some $\omega_1, \omega_2 \in \Omega$: if $\omega_1$ and $\omega_2$ do have at least one joining vertex, then according to theorem \ref{the:GeneralizedHeapConjunctionTheorem} the resulting heap graph is well-defined, otherwise the result is $\underline{false}$ (meaning $\omega_1$ and $\omega_2$ are disjoint). This way we have just shown that $\Omega$ is closed under $\circ$ and that a ''\textit{meaningful}`` heap graph conjunct may be obtained by connectible heap graphs. The connection is established successively.

Second, associativity needs to be demonstrated, namely that $m_1 \circ (m_2 \circ m_3) = (m_1 \circ m_2) \circ m_3$ holds. When looking at Fig. \ref{fig:HeapGraphConjBeforeAfter} we can immediately see the validity of both directions of the equation, because it does not matter whether $a$ and $b$ are joined first, or $a$ is connected to connected $b \circ c$, since the joining vertex of $b$ remains invariant when altering the connect ordering.

Now $G$ forms a semi-group, we still need to show there always exists some neutral element $\varepsilon$, so (iii) holds. This follows, however, immediately from the generalised heap theorem \ref{the:GeneralizedHeapConjunctionTheorem}.
\end{proof}

\textbf{Remark:} From (i) follows that $c \not \in b \wedge c \neq a$: $a \mapsto b \ \circ \ c \mapsto d \equiv \underline{false}$, and that $a\mapsto b \ \circ \ a \mapsto d \equiv \underline{false}$ holds. Furthermore, it is intuitively clear that connecting two heap graphs may be done using different joining vertices, regardless of which joints are connected first. The resulting heap graph shall be the same due to \textit{confluence}, due to (ii) and, moreover, due to the property being demonstrated later in definition \ref{lem:HeapConjunctionGroupProperty}.

\textbf{Remark:} Closeness (i) demonstrates the non\--re\-pe\-ti\-tiveness property of a Substructural Logic (the Separation Logic as still to be shown later) remains.

\begin{theorem}
$G=(\Omega, \circ)$ is an Abelian group.
\label{lem:HeapConjunctionGroupProperty}
\end{theorem}
\begin{proof}
Due to lemma \ref{lem:HeapConjunctionMonoid}, $G$ is a monoid; hence we still need to show (i) the existence of an inverse element for every heap graph, s.t.
\begin{eqnarray}
\forall \omega \in \Omega. \exists \omega^{-1} \in \Omega: \omega \circ \omega^{-1} = \omega^{-1} \circ \omega = \varepsilon
\label{eqn:InverseExists}
\end{eqnarray}
and (ii) $\circ$ is commutative.

Let us start to prove the induction with (ii): for the base case ''$loc_1 \mapsto var_1 \circ loc_2 \mapsto var_2 = loc_2 \mapsto var_2 \circ loc_1 \mapsto var_1$`` of definition \ref{def:HeapTermDefinition} the equivalence holds obviously. The inductive case holds also as long as the conditions on $\circ$ are obeyed, the proof induction can be deduced from Fig. \ref{fig:HeapGraphConjBeforeAfter}, implying commutativity holds for arbitrarily connected heaps. However, when it comes to abstract predicate, the boundaries of a $\circ$-connected heap graph term may vanish. This needs to be taken into consideration by whoever writes the specification.

Now, let us proceed with (i). We do have the problem of finding an inverse for whatever heap we get. The question what it means particularly raises immediately. If we think about natural numbers as operating carrier set and an attempt to invert addition, we factually introduce subtraction on integers. The same happens to complex numbers as an extension of real numbers. Nobody really is not able to count imaginary numbers in practice. Nevertheless, this extension seem to simplify basic calculations significantly in applications. So, why not assume for the moment and postulate equation (\ref{eqn:InverseExists}) right until found different?

And so, what does it \textit{intuitively} mean ''\textit{to negate a heap}`` ?  One could think of negating a points-to predicate affects only the state or that there is just no such heap reference. However, it does not really \textit{undo} a heap reference after all. A hypothetical ''\textit{not-points-to}`` requires primarily some kind of a ''\textit{transcendental heap removal}`` -- at the first glance this may indeed sound like a delicate problem, because up to now we were only specifying what is in the heap. We shall also be able to specify what is not in, but we had no chance whatsoever to specify a predicate which states some heap must be removed ''\textit{somehow}``. Let us not bother about it too much for the moment and focus instead only on equation (\ref{eqn:InverseExists}).
 What this equation actually states is a \textit{negated points-to} $a \not \mapsto b$ relationship, or more generalised some \textit{negated heap} $H^{-1}$, s.t. by convention $a \mapsto b \circ a \not \mapsto b = \underline{emp}$ and $a \not \mapsto b \circ a \mapsto b = \underline{emp}$, and more general: $H \circ H^{-1} = H^{-1} \circ H = \underline{emp}$. This means $\omega \circ \omega^{-1}$ removes a heap, and in fact it is an edge removal in addition to an optional heap graph vertex removal in case there are no more edges going to/leaving from the corresponding heap vertex. It is now easy to see why $H \circ H^{-1} \circ H \equiv H$ holds. For demonstration let us have a look at Fig. \ref{fig:HeapGraphInversion}. The heap states before inversion $d\mapsto a \ \circ \ a\mapsto b \ \circ \ c \mapsto b$, when applied $\circ (a\mapsto b)^{-1}$ we obtain $d\mapsto a \ \circ \ a\mapsto b \ \circ \ c \mapsto b \ \circ \ (a\mapsto b)^{-1}$ equals $d\mapsto a \ \circ \ a\mapsto b\ \circ \ (a\mapsto b)^{-1} \ \circ \ c \mapsto b$ equals $d\mapsto a \ \circ \ c \mapsto b$, which may not occur quite plausible at the first view yet, because both pointers do not interfere. Therefore it is required to perform one generic step. 
 
\textbf{Canonization step I}: If a bridge is removed from between sub-heap graphs, then the conjunction needs to be replaced by a heap disjunction.
  
For the example a bridge is detected between $a$ and $b$, so $\circ$ may be substituted by $\|$ in the remaining heap term, and so the result appears plausible again. But for sake of completeness heap graph vertices may need to be removed even completely. This becomes urgent especially when it comes to object field locations.
 
\textbf{Canonization step II:} Remove some heap graph vertex $a$ entirely whenever there are no more references to it.
 
Applying those two canonization steps keeps the chosen model sound and confluent (proof skipped here).

 \textbf{Remark:} We did not mention heap generalisations explicitly in the previous paragraphs, although it remains to the reader to prove correctness of $H \circ H^{-1} \equiv \underline{emp}$ (proof by induction over $\circ$, use the fact that $(g_1 \circ g_2)^{-1} \equiv g_1^{-1} \circ g_2^{-1}$ holds, so there exists a homomorphism for $.^{-1}$ w.r.t. $\circ$, refer to  lemma \ref{lem:HeapInversionHomomorphism}).
 
 \textbf{Convention:} Condition (i) implies particularly $\underline{emp} \circ \underline{emp}^{-1} \equiv \underline{emp}$, because we agree on $\underline{emp}^{-1} \equiv \underline{emp}$.
\end{proof}

\begin{figure}[t]
\begin{tabular}{lllll}
 \xymatrix{
   a \ar[r] &  b\\
   d \ar[u] & c \ar[u]
 }
& \parbox[t]{1.7cm}{\begin{center}conjunct:\\$(a \mapsto b)^{-1}$\\ $\Longrightarrow$\end{center}} &
 \xymatrix{
   a \ar@{..>}@/^/[r] &  b\\
   d \ar[u] & c \ar[u]
 }
 &
  $\Leftrightarrow$
 &
 \xymatrix{
   a  &  b\\
   d \ar[u] & c \ar[u]
 }\\
\end{tabular}
\caption{Heap graph before and after inversion}
\label{fig:HeapGraphInversion}
\end{figure}
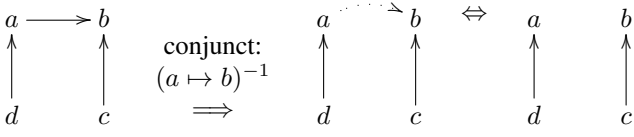

\textbf{Subsumption:} $H \circ a\mapsto b \circ (a\mapsto b)^{-1}$ denotes:

\begin{enumerate}
 \item unlink edge between $a$ and $b$
 \item unlink/remove $a$ if there are no more uses in $H$
 \item unlink/remove $b$ as well if no more uses in $H$
\end{enumerate}

The group properties allow us to define equalities for the separated heap theory. This will allow us, for instance, to define arithmetic equalities, which applied will cause faster convergence to a normalised heap representation. It will lower the risk of highly bloated verification rules and conclusively it will lead us to a smaller logic in combination with partial specification (see section \ref{sec:PartialSpec}). 
Future work may include heap arithmetics to be implemented by satisfiability-modulo-theory solvers, which will be integrated to the verification process. This approach does not only sound promising, but in fact it was successfully proven concept in several different areas, particularly for bloated and notoriously incomplete Hoare logics.

\begin{lemma}
 $(g_1 \circ g_2)^{-1} \equiv g_1^{-1} \circ g_2^{-1}$ holds for any heaplets $g_1$ and $g_2$.
\label{lem:HeapInversionHomomorphism}
\end{lemma}

\begin{proof}
 Lets generalise this lemma, let $G=g_1\circ g_2\circ \cdots \circ g_n$, we need to show $G \circ G^{-1} = \underline{emp}$. This can be shown by induction over $n$. In the base case ($n=1$) we have $g_1\circ g_1^{-1} \equiv \underline{emp}$,  which holds because of the existence of an inverse. 
 For the inductive step let $G=\underbrace{(g_1\circ g_2 \circ \cdots \circ g_k)}_{G_k} \circ g_{k+1}$, then $G\circ G^{-1} = (G_k \circ g_{k+1}) \circ (G_k \circ g_{k+1})^{-1}$ denotes in the inverse part a graph extension. The right part of this equation equals $\underbrace{G_k\circ G_k^{-1}}_{\underline{emp}}  \circ \underbrace{g_{k+1} \circ g_{k+1}^{-1}}_{\underline{emp}} \ \equiv \ \underline{emp}$ (because of the inductive inversion property).
\end{proof}


\begin{definition}
\textit{Heap disjunction} $H \ \| \ a \mapsto b$ defines heap $H$ and heaplet $a \mapsto b$ which do not interfere, iff $G_H$ is the heap graph of $H$, $G_H=(V,E)$, and for all edges $(\_,a) \not \in E$ and there exists no path from $b$ to $H$, and there is no path back from $H$ to $a$.
\end{definition}

That is why $x.b \ \| \ x.c$ does not hold for any object $x$ with fields $b$ and $c$, if there exists at least one common vertex on any path from $x.b$ or from $x.c$.

Let $\Sigma = X_0 \| X_1 \| \cdots \| X_n$ with $n>0$ and $X_j$ is of form $x_j \mapsto y_j$, then $\Sigma = \Sigma_0 \ \| \ a_0 \mapsto b_0 \Leftrightarrow \forall (a_j \mapsto b_j) \in \Sigma_0: a_j \neq a_0 \wedge b_j \neq b_0$.

\begin{theorem}
$G=(\Omega, \|)$ is a monoid and a group, if $\Omega$ is the set of heap graphs and $\|$ denotes heap disjunction.
\end{theorem}
\begin{proof}
In analogy to the previous lemma, first of all, $\forall m_1,m_2 \in \Omega: m_1 \| m_2$, iff $m_1$ and $m_2$ have no common joint, which is the case whenever there is no path from $m_1$ to $m_2$, and there is not even an indirect heap graph surrounding both $m_1$ and $m_2$. If $m_1$ and $m_2$ are different, then $m_1 \| m_2$ must be a valid heap $\in \Omega$ again, because $m_1$ is from a different heap graph part than $m_2$, and vice versa, so closeness holds.
Associativity holds obviously.
$\underline{emp}$ may be chosen as neutral element, so $\underline{emp} \| m_1 = m_1 \| \underline{emp} = m_1$, by default let $\underline{emp} \| \underline{emp} = \underline{emp}$ hold.
Last, we agree on the convention $s \| s^{-1} = s^{-1} \| s = \underline{emp}$, which behaves similar to $\circ$. Heaps in general obey this rule.
\end{proof}

The heap-wise conjunction and disjunction may be expressed as following:

\begin{tabular}{ll}
 \inference[$\circ_{[B,C]}$]{U \circ B \ \| \ C}{U\circ B\circ C} &
 \inference[$\|_{[B,C]}$]{U\circ B\circ C}{U \circ B \ \| \ C}
\end{tabular}

\begin{eqnarray}
\|_{[B,C]} ; \circ_{[B,C]} ; \|_{[B,C]} \equiv \|_{[B,C]}
\label{eqn:HeapDisjunctionInvariant}\\
\circ_{[B,C]} ; \|_{[B,C]} ; \circ_{[B,C]} \equiv \circ_{[B,C]}
\label{eqn:HeapConjunctionInvariant}
\end{eqnarray}

The operations $\|$ and $\circ$ are dual and self-inverse as can be seen from equations (\ref{eqn:HeapDisjunctionInvariant}) and (\ref{eqn:HeapConjunctionInvariant}), where '';`` is the statement sequentializer. The equations do hold (direct proof, skipped here), because of its self-inverse character and due to the assertion that both specified heap vertices $B$ and $C$, in fact, exist.

\begin{theorem}
Distributivity holds for $\forall a,b,c \in \Omega$ for $\circ$ and $\|$:

\begin{tabular}{ll}
 (i) & $a \circ (b \| c) = (a \circ b) \| (a \circ c)$\\
 (ii) & $(b \| c) \circ a = (b \circ a) \| (c \circ a)$ 
\end{tabular}
\label{lem:DistributivityForConjDisj}
\end{theorem}

\begin{proof}
 (direct proof, skipped, take note of Fig. \ref{fig:HeapGraphPoset}).
\end{proof}

\textbf{Remark:} Since the neutral element for both operations, $\circ$ and $\|$, is $\underline{emp}$, there cannot be defined a field over both operations, although lemma \ref{lem:HeapConjunctionMonoid}, theorem \ref{lem:HeapConjunctionGroupProperty} and  \ref{lem:DistributivityForConjDisj} hold, and $\Omega$ is finite. Particular heaps would be finite. All operations applied to finite heaps would be finite again.

\textbf{Remark:} In analogy to logical conjuncts $\wedge$ and $\vee$ a $\|$-normalform exists when the previous equalities are applied. Lemma \ref{lem:HeapInversionHomomorphism} can be applied for the inversion of generalised heaps.\\

In order to optimize reasoning by minimizing graph size, $\|$ should be moved upwards at most in heap terms, e.g., by applying distributivity rules or by reordering points-to heaps so that the left-hand sides are ordered in lexicographical order by its location. The motivation behind this is, for instance, to optimize incremental verification, so only affected heaps may require re-calculations.

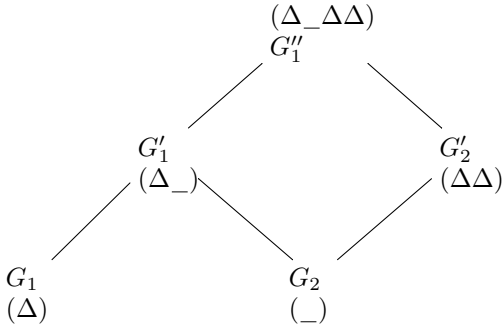
\begin{figure}[t]
\begin{displaymath}
 \xymatrix{
      &  & \parbox[t]{12mm}{$(\Delta\_\Delta\Delta)\\\ \ G_1''$} & \\
      & \parbox[t]{7mm}{$G_1'\\(\Delta\_)$} \ar@{-}[ru] &  & \parbox[t]{7mm}{$G_2'\\(\Delta\Delta)$} \ar@{-}[lu] \\
     \parbox[t]{7mm}{$G_1\\(\Delta)$}  \ar@{-}[ru] &  & \parbox[t]{7mm}{$G_2\\(\_)$} \ar@{-}[lu] \ar@{-}[ru]
 }
\end{displaymath}
\caption{A partially-ordered set (poset) can be defined for heap graphs under inclusion, the join operator is $\circ$.}
\label{fig:HeapGraphPoset}
\end{figure}

As seen in Fig. \ref{fig:HeapGraphPoset} a partially-ordered set can always be defined over inclusion of heap graphs. Despite there might be an infimum defined as $\underline{emp}$ and some always existing supremum, the complete heap graph, the structure is still not a (complete) lattice due to non-holding connective properties w.r.t. $\|$ as meet for absorption. The poset $G$ from Fig. \ref{fig:HeapGraphPoset} contains $\{G_1,G_2,G_1',G_2',G_1''\}$ and can be ordered by the following ascending chains $G_1 \sqsubseteq G_1' \sqsubseteq G_1''$, $G_2 \sqsubseteq G_1'$ and $G_2 \sqsubseteq G_2' \sqsubseteq G_1''$. The supremum is  $G_1'', inf(G)=\underline{emp}$, where $\sqsubseteq$ shall be defined as the heap sub-graph relationship. Obviously, if two heaps are not disjoint (this denotes $\underline{emp}$ because of definition \ref{def:HeapConjunctionDefinition}) they may always be connected with each other in the corresponding Hasse-diagram. This join is always valid because of (1st contradiction) $a\mapsto b \ \circ \ a\mapsto d$ may not occur after a single heap conjunct, or any composite heap in general (2nd contradiction) $a\mapsto b \| b\mapsto d$ contradicts the definition of $\|$. However, it needs to be taken into consideration that the inclusion-ordering mentioned may be destroyed by applying inverse heaps if used arbitrarily (compare with previous section), but those may usually, at least at the moment, be used only in cases where a difference between expected and actual heap graphs needs to be calculated rather than a desired heap graph specification, so the locality property mentioned from section \ref{section:Introduction} remains untouched.


\section{Partial Heap Specification}
\label{sec:PartialSpec}

Having said earlier after definition \ref{def:HeapTermDefinition} class objects may be considered as containers of fields $obj.f_1 \mapsto .. \circ obj.f_2 \mapsto .. \circ obj.f_n \mapsto .. $, all fields constitute a scoped heap (in analogy to single points-to local variables among abstract predicates). Since class object fields may not exist independently from other fields of the same class, they must by convention be $\circ$-conjuncted. In contrast to locals, object fields too require a possibility to specify only parts, hence constants from definition \ref{def:HeapTermExtendedDefinition} are parameterised to $\underline{true}(obj)$ or $\underline{false}(obj)$.  Abstract predicates modularise specifications, they can particularly be used to specify objects from unrelated objects and locals.
W.r.t. the proposed stack-based implementation of a ``$a \circ a$''-decider incoming and outgoing terms for abstract predicates may be traced in order to skip re-verficiation of unaffected parts.

\begin{definition}
 A partial heap specification $\underline{t}(o)$ for some object $o$ is defined as a $\circ$-conjunction of all remaining fields, possibly none, that are not specified until $\underline{t}$ is used and unfold in the current object scope. When $\underline{t}$ is used all actual fields are unfold into the surrounding heap specification, which are not yet been specified in terms of the current scope of $o$.
\end{definition}

\begin{example}
Lets say object $a$ has three fields $f_1$, $g_1$ and $g_2$, and \comp{}{} denotes some (implicit) heap term denotation of type $(ET \rightarrow ET) \rightarrow \{true, false\}$, where the first extended heap is supposed to be the expected and the second extended heap is supposed to be the actual heap term then
 \begin{eqnarray*}
   & \comp{$a.f_1 \mapsto x\circ \underline{true}(a)$}{} = & \comp{$a.f_1 \circ a.g_1 \circ a.g_2$}{}\\   
   = & \comp{$\underline{true}(a) \circ a.f_1 \mapsto x$}{} \neq & \comp{$p(a) \circ a.f_1 \mapsto x$}{}
 \end{eqnarray*} 
where $p$ is some abstract predicate denoting $\underline{true}(a)$. However, \comp{$a.f_1 \mapsto x \circ p(a)$}{} would denote equality, because the stack-oriented recognition finds all remaining fields even when obfuscated beneath several predicate levels. \comp{}{} is a homomorphic mapping regarding $\circ$.
\end{example}

\begin{example}
 \comp{$\underline{true}(a) \circ \underline{true}(a)$}{} = \comp{$a.f_1 \circ a.g_1 \circ a.g_2$}{} $\circ$ \comp{$\underline{true}(a)$}{} = \comp{$a.f_1 \circ a.g_1 \circ a.g_2$}{} $\circ \ \underline{emp}(a)$.
\end{example}


\section{Discussions}

By exchanging one ambiguous spatial heap operator by two strict operations, the initial and core properties of a Separation Logic did not change essentially, except an unbound and arbitrary heap inversion as discussed in section \ref{sec:PartialSpec}. The introduction of a strict normalform allows a linear and local analysis of the heap terms without an eager comparison of still non-matched conjuncts.

As mentioned in section \ref{sec:ConjunctionAndDisjunction} there arises the question of inconsistency, whether in remote parts of the same specification, for instance, somewhere up or down relative to the current abstract predicate calling stack, there are in fact two identical heaplets or if there are any two pointers with the same location multiple times. The reason beyond is non-repetitiveness.

This problem may be resolved in general only dynamically during the verification due to the undecidability of abstract predicates due to the undecidability of the Halting-problem. Hence a stack-based analysis for processing abstract predicates is needed very similar to the operational semantics provided by Warren \cite{warren83} including processing of symbols and back-references to the stack parameters, except that abstract predicates require an adapted reasoning control (see \cite{haberland14}).
Applying Warren's approach to strict heap conjunction and disjunction will decide $\forall a.a \circ a$. Because of ungrounded symbols, a memoizer may not cope with global states in abstract predicates.

Prolog is a general purpose logical programming language \cite{Sterling94}. We strongly believe Prolog may be used to reason about extended heap terms and abstract predicates in order to resolve key heap verification problems, such as expressibility restrictions \cite{haberland14}.
One advantage of Prolog predicates over classic \textit{one-way} functions (as used in \cite{parkinsonThesis05}, for instance) is that input and output terms may be considered as relation, unioning exponentially many different combinations of input and output vectors, skipping only those combinations where a relation is not defined. This is often the case when the corresponding one-way function is either non-invertible, contains cuts or arithmetic evaluations (e.g. by using the built-in predicate \texttt{is}). There must be a strong correlation between input and output vectors, s.t. a bijective mapping exists between both vectors for the most common definition.
Besides, Prolog predicates containing cuts may always be rewritten w.l.o.g. cut-free. Predicates containing \texttt{is} may be rewritten without as ground term, e.g. a Church number, as long as it can be represented as unifiable term, which is feasible even if not too elegantly.

The ``\textit{Object Constraint Language}'' (OCL) \cite{oclspec} is a specification language for class-instantiated objects in companion with the ``\textit{Unified Modeling Language}''. It implements a fragment of the predicate logic, it supports some quantification of variables and supports collection types and ad-hoc polymorphism by sub-classing. It allows life-cycle specification of objects and class methods. However, OCL does not know of pointer constraints like aliasing.
If pointer constraints were added to OCL and propagated down to code generation when compiling user code, for instance, then code performance could significantly improve (compare with \cite{haberland14}).
In combination with abstract predicates the proposed modification of this paper may be used as proposition for an update of the recent OCL recommendation w.r.t. the intrinsic points-to model, particularly referring here to conventions \ref{conv:RestrictedObjects}, and definitions \ref{def:HeapTermDefinition} and \ref{def:HeapTermExtendedDefinition}.
Presumably, this would also raise expressibility, abstraction and higher modularity. Previous attempts to demonstrate applications of Separation Logic to Design Patterns can be found for instance in \cite{parkinsonThesis05}.


\section{Conclusions}

At the beginning, the problem of verifying dynamic heap was introduced. Related problems and the benefit of the heap separation model were provided. The description of points-to assertions corresponds to heap-manipulating program statements -- in the chosen model the generated heap graph is described edge-wise. The problem with the $\star$-operator was that it may be used syntactically and semantically in two different ways: for heap disjunction, but also for heap conjunction. This caused the described issues. The introduction of two strict heap operations allows heap terms to be interpreted simple. In addition to this, partial object specifications make it promising to understand and control better completeness w.r.t. incoming heaps in heap formulae. Properties of both operations were investigated and found restrictive, but still flexible enough for expressing arbitrary heap graphs. The integration of derived rules to a SMT-solver requires further research. Finally an extension of the current OCL was proposed.


\section*{Acknowledgement}
 
Some parts of this paper have been prepared within the
scope of project part of the state plan of the Board of Education and Science
of Russia (task \# 2.136.2014/K).

\end{document}